\documentclass{article}
\usepackage{amsfonts}
\usepackage{epsfig}
\usepackage{amssymb,amsmath,bm}
\usepackage{subfigure}
\usepackage{amsthm}
\usepackage{mathrsfs}
\usepackage[section]{algorithm}
\usepackage{fancyhdr}
\theoremstyle{plain}
\newtheorem*{theorem*}{Theorem}

\newtheorem{tvd1}{Theorem}
\newtheorem{res}[tvd1]{Theorem}
\newtheorem{prop}[tvd1]{Proposition}
\begin{document}
%\myheads{A. Jasra, D. A.~Stephens and C. C.~Holmes} {Population-based reversible jump}

\bigskip

\begin{center}

{\Large \textbf{Population-based reversible jump \\
Markov chain Monte Carlo}}

\bigskip

AJAY JASR$\textrm{A}^{\star}$, DAVID A. STEPHEN$\textrm{S}^{*}$ \& CHRISTOPHER C. HOLME$\textrm{S}^{\dagger}$\\[0pt]

\emph{${}^{\star}D$epartment of Mathematics, Imperial College London, SW7 2AZ, London, UK%
}\\[0pt]

\emph{${}^{*}D$epartment of Mathematics and Statistics, McGill University, H3A 2K6, Montreal, CA%
}\\[0pt]

\emph{${}^{\dagger}D$epartment of Statistics, University of Oxford, OX1 3TG, Oxford, UK%
}\\[0pt]
\end{center}

\begin{abstract}
In this paper we present an extension of population-based Markov
chain Monte Carlo (MCMC) to the trans-dimensional case. One of
the main challenges in MCMC-based inference is that of
simulating from high and trans-dimensional target measures.
In such cases, MCMC methods may not adequately traverse the support of the
target; the simulation results will be unreliable. We develop population methods to deal with such
problems, and give a result proving the uniform ergodicity of
these population algorithms, under mild assumptions.  
This result is used to demonstrate the superiority, in terms of
convergence rate, of a population transition kernel over a
reversible jump sampler for a Bayesian variable selection
problem. We also give an example of a population algorithm for a
Bayesian multivariate mixture model with an unknown number of
components. This is applied to gene expression data of 1000 data
points in six dimensions and it is demonstrated that our algorithm out
performs some competing Markov chain samplers.\\
\noindent {\emph{Some key words}: Population Monte
Carlo, Uniform ergodicity, Bayesian variable selection, Mixture
models, Gene expression data}
\end{abstract}

\bigskip

\section{Introduction}

The Metropolis-Hastings algorithm (Metropolis et al.~1953;
Hastings, 1970) and its adaptation to the trans-dimensional case
(Green, 1995) has provided a
method to simulate from complex probability measures in high
dimensions. This has facilitated the application of 
(particularly Bayesian)
complicated statistical models which could not otherwise be
fitted.

We consider the problem of simulating from
a probability measure $\pi(x)\lambda(dx)$ defined on measurable
space $(E,\mathcal{E})$ (where $\lambda$ is a $\sigma-$finite
measure on $\mathcal{E}$), with $\pi$ known pointwise, at least
up to a normalizing constant.  This is
achieved (in the context of MCMC) by simulating an ergodic Markov chain $\{X_n\}_{n\geq
0}$, with kernel $K:E\times\mathcal{E} \rightarrow[0,1]$, of
stationary distribution $\pi$. The Markov chain can be used, for
example, to estimate expectations of $\pi-$integrable functions.
In this
article we focus on reversible jump Markov chain Monte Carlo
(RJMCMC), where the state space is a union of subspaces of
differing dimension, that is, where
$E=\bigcup_{k\in\mathcal{K}}\big(\{k\} \times E^{k}\big)$,
$\mathcal{K}\subseteq\mathbb{N}$,
$E^{k}\subseteq\mathbb{R}^{k}$.

In statistical terms, $\pi$ will often be a Bayesian posterior
distribution which is normally known pointwise, up to a normalizing constant.
For example, applications of RJMCMC include classification and regression
(Denison et al.,~2002) and mixture modelling (Richardson and Green,
1997), with particular emphasis on model determination. However, in many
examples, the naive (vanilla) RJMCMC sampler can fail to
move around the support of the target in a feasible computation
time; see Brooks et al.~(2003) for examples.

To deal with these problems, several MCMC approaches have been
suggested, including auxiliary variable methods (Brooks et
al.,~2003) and tempered transitions (Jennison et al.,~2003):
Green (2003b) provides a recent review.  One approach used for
difficult sampling problems in fixed dimensional spaces, that
have not been widely used in the variable dimension case, relies
on population-based MCMC methods (Liang \& Wong, 2001; Liu,
2001). It is straightforward, conceptually, to extend this
approach to the variable dimension case.  Such an extension is
the focus of this paper; we study the potential theoretical and
practical advantages of population approaches over standard MCMC
methods. We remark that methods other than MCMC may be
used for difficult simulation problems, such as sequential Monte
Carlo (e.g.~Del Moral et al.,~(2006)), but such methods are not
the focus of this paper.

\subsection{Population-based Markov chain Monte Carlo}

Population-based MCMC operates by embedding the target into a
sequence of related probability measures
$\{\pi_i\}_{i\in\mathbb{T}_N}$, $\mathbb{T}_N:=\{1,\dots,N\}$,
and simulating the $N$ parallel chains (the
\textit{population}), as in parallel tempering (Geyer, 1991;
Hukushima \& Nemoto, 1996). In addition, the chains are allowed
to interact via various crossover moves; a summary is given in
Section 3 - see Liu (2001) for an extensive review.

The main advantage of population-based simulation over other
methods is the fact that the population simultaneously represents
many properties of the target distribution. This is particularly
useful in trans-dimensional simulation, where it can be
difficult to construct efficient dimension-changing proposals.
Green (2003b) notes that some MCMC methods retain information
about which states have been visited, for example, the product
space approach (Carlin \& Chib, 1995; Godsill, 2001), whilst
standard reversible jump MCMC does not retain this information.
There are advantages to both approaches.  The first approach can
provide an improved mode-jumping property; that is, the ability to jump a large
number of dimensions that would take a substantial time under
the standard approach.  The second approach has greater capacity
to discover new states that are consistent with the target.  It
is clear that an algorithm which can combine these properties is
likely to provide an improvement over both methods; the objective of this paper is to
construct such an algorithm, and to investigate its theoretical
properties.

\subsection{Contribution and Structure of Paper}

Theoretical aspects of population-based MCMC have been rarely
considered (see, however, Madras and Zheng (2003)), and
therefore the improvements, if any, offered by population
methods over standard MCMC have not been fully established.  In
this article we present a result which ensures (under fairly
mild conditions, including that the density $\pi$ is upper bounded)
the uniform ergodicity of a population transition kernel, and
allows the construction of population algorithms which are
preferable, in theory, to their single chain counterparts. The
result can help to illuminate, for small $N$, why population
algorithms can work well in practice.

We also demonstrate that a particular population-based
kernel - the \textit{exchange kernel} (see Sections
\ref{mainres} and \ref{exch}) which is fundamental to swapping
information in population MCMC - can improve convergence
properties (over parallel MCMC, in which $N$ independent
identical Markov chains are run) under strong assumptions; this has not, to our
knowledge, been established previously. 
Our results are substantiated with an example in Bayesian variable selection, and
illustrate the use of population MCMC methods in a mixture
problem.

%On investigating difficult sampling problems, one is quickly drawn
%to the conclusion that no one single method is guaranteed to work.
%However,
Population methods naturally accommodate multiple MCMC
strategies which improve the ability of the sampler to mix across
the state space. In our main example we show how to combine the
methods of parallel chains (Geyer, 1991), tempering (e.g.~Geyer
\& Thompson, 1995), snooker algorithms (Gilks, et.~al., 1994),
constrained sampling (e.g.~the unpublished technical report of Atachd\'e
\& Liu (2004)) and delayed rejection (Green \& Mira, 2001).  We
believe that although such methods may not perform adequately
(individually) together they can provide a superior MCMC
sampler.

This paper is organised as follows. In Section 2 we provide an
illustrative example related to the clustering of gene
expression data via Bayesian multivariate mixture models. In
Section 3 we introduce population-based reversible jump. In
Section 4 we give some theoretical results that can indicate why
population methods can perform well. In Section 5 we present
population MCMC moves for the mixture model example. In Section
6 we provide a comparison of vanilla, simulated tempering (Geyer
\& Thompson, 1995) and population samplers for the mixture model
example. In Section 7 we conclude with a discussion, detailing
extensions to our approach.

\subsection{Notation}

The notation and mathematical objects that are adopted in the paper
are summarized here.

A measurable space is denoted $(E,\mathcal{E})$: throughout
this article $\mathcal{E}$ is a countably generated
$\sigma-$algebra. The product $\sigma-$algebra is written
$\mathcal{E}^N:=\mathcal{E}\otimes\cdots\otimes \mathcal{E}$
(product $N-$times). We use
$\delta_x(dx')$ to represent Dirac measure

For two probability measures, $\lambda_1$ and $\lambda_2$, on
$\mathcal{E}$, the total variation distance is written
$\|\lambda_1-\lambda_2\|_{TV} :=
\sup_{A\in\mathcal{E}}|\lambda_1(A)-\lambda_2(A)|$. The
Dobrushin coefficient (Dobrushin, 1956) of a Markov kernel $K$
is denoted $\beta(K):=\sup_{(x,y)\in
E^2}\|K(x,\cdot)-K(y\cdot)\|_{TV}$. The composition of two
Markov kernels, $K$ and $P$ is written $K\circ
P(x,dy):=\int_{E}K(x,du)P(u,dy)$, except if $K=P$ where
$K^2$ is used. Particular Markov kernels of interest are the product
and mixture kernels that combine component kernels $K_i$ in a
multiplicative and additive fashion, which will be denoted
\[
\prod_i K_i \qquad \sum_i \tau_i K_i
\]
for mixture weights $\tau_i$.  Given a probability
measure $\lambda$ and Markov kernel $K$, the standard
notation $\lambda K(dy):= \int_{E}\lambda(dx)K(x,dy)$ is adopted. For a
bounded measurable function $f$, the oscillations are written
$\textrm{osc}(f)=\sup_{(x,y)\in E^2}|f(x)-f(y)|$.

Finally, we use the vector notation $x_{1:N}:= (x_1,\dots,x_N)$
and denote a vector with its $i^{th}$ element missing as $x_{-i}$ and
with only its $i^{th}$ and $j^{th}$ elements as $x_{i,j}$. 
Also, $\mathbb{T}_{l:p}:=\{l,\dots,p\}$ for any $l\leq p$, $(l,p)\in\mathbb{Z}^2$.

\section{An Illustrative Example: Finite Mixture Modelling}
\label{IllustrativeExample}

Mixture models are typically used to model heterogeneous data,
or as a simple means of density estimation; see McLachlan \&
Peel (2000) for an overview. Bayesian analysis using mixtures
with an unknown number of components has only fairly recently
been implemented (see Richardson \& Green (1997) and in the
multivariate context Stephens (2000) and Dellaportas \&
Papageorgiou (2006)).

In this section, we describe the finite mixture model adopted and a
motivating example, in which standard MCMC methods do
not perform adequately.

\subsection{Model}
 Let $y_{1:n}$ denote observed data
that lie on support $y_{i}\in Y\subseteq \mathbb{R}^{r}$,
$i=1,\dots,n$. We assume that the $y_{i}$ are i.i.d with density:
\begin{eqnarray*}
p(y_{i}|\eta_{1:k},w_{1:k}) & = &
\sum_{j=1}^{k}w_{j}f(y_{i};\eta_{j})
\end{eqnarray*}
where $\eta_{1:k}$ are component specific parameters, the weights
$w_{1:k}$ are such that $\sum_{j=1}^{k}w_{j}=1$, $w_{j}\geq
0~\forall \; j$, $p$ denotes an arbitrary probability
density function and $f$ is the component density. For
our model, we restrict ourselves to the case of
multivariate $t$, $\mathcal{T}_{r}(\mu,\Lambda,s)$, where
$(\mu,\Lambda)$ are the location and covariance parameters and $s$
is the degrees of freedom.

In specifying the prior distributions, we follow Stephens (2000);
the component mean vectors are taken to be independent
$\mathcal{N}_{r}(\xi,\kappa^{-1})$ (multivariate normal distribution), and the $\Lambda_{j}$ are
independently $\mathcal{IW}_{r}(2\alpha,2\Psi)$, where
$\mathcal{IW}_{r}(\cdot,\cdot)$ is the inverse Wishart distribution.
The following hierarchical structure is adopted:
\begin{eqnarray*}
\Psi & \sim & \mathcal{W}_{r}(2g,(2h)^{-1}) ~~~~~\text{where}~\mathcal{W}_{r}(\cdot,\cdot)~\text{is the Wishart distribution}\\
w_{1:k-1}|k & \sim & \mathcal{D}(\delta) ~~~~~~~~~~~~~~~\text{where}~ \mathcal{D}(\cdot)~\text{is the symmetric Dirichlet distribution}\\
k & \sim & \mathcal{U}_{\{1,\dots,k_{\textrm{max}}\}}
~~~~~\text{where}~\mathcal{U}_{S}~\text{is the discrete uniform
distribution on countable set}~{S}.
\end{eqnarray*}
When using the multivariate $t$-distribution, the degrees of
freedom are assumed known. Thus, our prior is:
\begin{eqnarray*}
p(w_{1:k-1},\mu_{1:k},\Lambda_{1:k},\Psi,k) & = & \bigg[\prod_{j=1}^{k}
p(\mu_{j})p(\Lambda_{j}|\Psi)\bigg]
p(\Psi)p(w_{1:k-1}|k)p(k).
\end{eqnarray*}

\subsection{Data Processing and Prior Distributions}

For our example we consider the problem of clustering gene
expression data ( e.g.~Heard et al.~(2006)). The data consist of
the relative level of gene expression - a measure of genetic
activity - for $n=4221$ genes of the parasite \emph{Plasmodium}
measured at $r=46$ time points across a 48 hour portion of the
parasite life cycle. The data are discussed in detail in Bozdech
et al.~(2003). Finding meaningful subgroupings of the data is an
important task for biological investigators.

Even with modern computing power, applying a fully Bayesian
analysis to such data is not practical. Therefore, 
the data is preprocessed the to reduce the $n\times r$-dimensional data
to $l\times q$ dimensions.  We
%(adopting an approximate Bayes approach), and
achieve this by adopting a $\mathsf{K}$-means partitioning approach to reduce
$n$ to $l$, and then principal components to reduce $r$ to $q$. We
selected $l=1000$ and $q=6$.

%

%To reduce $n$ to $l$, we adopt a simple hierarchical scheme, where

%we apply $l$-means clustering to the data, for some large $l$ and

%then attempt to cluster the means. In order to reduce $r$ to $q$, we

%take the first $q$ principal components of the means.

%To select $l$, we investigated the clustered genes for various $l$

%between 500 - 1500 (we used the default $l$-means algorithm in

%S-Plus). We found apparently sensible clusters, where the time

%series plots for the clusters had genes which were similar for most

%of these values, and chose $l=1000$. We then used the first six

%principal components of the cluster means, which account for 99$\%$

%of the variance of the data.

%

Prior hyperparameters are set in a similar way to Stephens (2000); we
set $\xi$ to be the midpoint of the observed data in
its corresponding dimension, and $\kappa$ is taken to
be $\textrm{diag}(1/R_{1}^{2},\dots,1/R_{q}^{2})$ where $R_{d}$ is the
range of the data in dimension $d\in\mathbb{T}_q$. Additionally $g=q/2$,
$\delta=1$ and $\alpha=\alpha'+(q+1)/2$, where $\alpha'=3$. Finally
$h$ is
$\textrm{diag}(100q/(2\alpha R_{1}^{2}),\dots,$ $100q/(2\alpha R_{q}^{2}))$.
For illustration a $t$-distribution with four degrees of freedom is used as the component density in our mixture model and set $k_{\textrm{max}}=20$.

\subsection{Performance of Vanilla Sampler}

\label{rjex}

A vanilla reversible jump sampler (outlined in Appendix 1) was
implemented for the data above. We ran the reversible jump
algorithm for 250000 sweeps from two different starting points,
which had different initial $k$. The C program was run on a
Pentium 4, 3 Ghz machine and took approximately three hours. We observed
extremely poor mixing, with all variable dimension
acceptance rates below 1$\%$. It can be seen that there appears to
be support for $k\in\mathbb{T}_{3:5}$ components, but also for
$k\in\mathbb{T}_{8:11}$. The main problem is that the vanilla
sampler cannot jump between these two modes.

The poor performance of the vanilla algorithm can be partly
attributed to the difference in dimensionality between different
mixture models.  To jump from a $k$ to a $k+1$ component mixture
model, we need to draw $1+q+q(q+1)/2=28$ random variables, so we
will need proposals/jump functions that are more tailored than
those used in the vanilla sampler.  In addition, if
between-model jumps are made infrequently, a within-model chain
spends a long time in a high density region of the
state space specific to that model that is consistent with the
data. This within-model adaptation renders a jump between models
even less likely, as high density regions in different models do
not necessarily correspond to each other in a straightforward
fashion.  This heuristic argument applies to local moves, and
whilst more global moves might be constructed (as noted by Green
(2003a), they are more likely to produce better mixing than the
local moves attempted here), dimension matching dictates that
this will be difficult to achieve efficiently.

\subsection{Alternative Algorithms}

Possible MCMC methods
%(we note that it is preferable that the existing code
%be utilized, instead of resorting to other simulation methods)
that might be used to deal with the problems encountered here
may be the auxiliary variable method of Brooks et al.~(2003),
but this will require a reasonable movement of the chain in the
first place. Constructing proposal distributions by creating an
approximation of the target in each dimension using fixed
dimensional MCMC (Hastie, 2005 (University of Bristol PhD
thesis)) is complicated by the label-switching problem (see
Jasra et al.~(2005) for a discussion).
%, see Jasra et al.~(2005a) for an example of this approach).
Delayed rejection (Green \& Mira, 2001) and tempered transitions
(Jennison et al.,~2003) often do not provide a general solution
to the problems highlighted by this example. For the former
method, insufficient information is learnt at the first stage
rejection to provide a significantly improved second stage
proposal, whereas for the latter, it often takes a large number
of intermediate simulations (e.g.~100) to provide a reasonable
proposal, but even so, the performance gain is not always
substantial.

%It is also apparent that the fixed dimensional

%simulation will not provide adequate mixing to provide good

%approximations.

%The problem of using the first method is that it

%requires that the sampler actually jumps somewhere to be effective.

%The second method will require good `pseudo priors', but this is in

%itself problematic.

\section{Population-Based Reversible Jump}

We now consider population-based reversible jump algorithms.  First
we give details of the population MCMC method, and then study the
theoretical properties of the algorithm, in particular its uniform
ergodicity.

\subsection{The Population MCMC Method}

Consider a sequence of probability measures
$\{\pi_{i}\}_{i \in\mathbb{T}_N}$,
each assumed to admit a density (also
denoted $\pi_i)$ with respect to $\lambda$ on $(E,\mathcal{E})$.
Denote the support of the $i^{th}$ density $E_i=\{(x_{k},k)\in E :
\pi_i(x_{k},k)>0\}$, $i\in\mathbb{T}_N$. Define probability measure
$\pi^*(d(x_{1:k_1},k_1,\dots,x_{1:k_N},k_N))$ on $(E^N,\mathcal{E}^N)$ as:
\begin{eqnarray*}
\pi^*(d(x_{1:k_1},k_1,\dots,x_{1:k_N},k_N)) & = & \bigg[\prod_{i=1}^{N}\pi_i(x_{1:k_i},k_i)\lambda(d(x_{1:k_i},k_i))\bigg].
\end{eqnarray*}
Our objective is to now construct an ergodic Markov kernel
$\widetilde{K}:E^N\times\mathcal{E}^N\rightarrow[0,1]$ with
stationary distribution $\pi^*$.

\subsection{The Structure of the Population}

The population approach proceeds as follows; we generate $N$
parallel (variable dimension) chains in order to explore the target
correctly.  For the remainder of the paper $\pi_{1}\equiv\pi$ and
the sequence of densities are taken $\pi_{i}\propto\pi^{\zeta_{i}}$,
$1=\zeta_{1}>\cdots>\zeta_{N}>0$, where $\{\zeta_{i}\}_{i \in\mathbb{T}_N}$ are inverse
temperature parameters; the collection $\{\zeta_i\}$ is referred to
as the \emph{temperature ladder}. This is the approach in Liang \&
Wong (2001), but other settings include taking $\zeta_{i}=1$ for
each $i$; see Del Moral et al.~(2006) for
further discussion.  We seek to use the extra information contained
in $N$ chains at different temperatures to allow large moves in
dimension of the chain of interest as well as allowing improved
performance in more local moves (within and between dimensions).

One of the main problems of parallel tempering (Geyer, 1991;
Hukushima \& Nemoto, 1996) is that only minimal interactions between
the chains are allowed.  Our approach differs as we will allow
$\widetilde{K}$ to include moves which use the entire
population, other than merely the exchange move (see Section
\ref{mainres}). Thus we seek a more population-based approach to
justify the increased cost in computation.

We now investigate some theoretical aspects of population
algorithms. As our results are not confined to the trans-dimensional
case, we drop the $k$ from the notation. That is, $E$ is a general
state space with associated countably generated $\sigma-$algebra
$\mathcal{E}$.

%\section{Some Theory for Population Samplers}
%We now consider some theory for population samplers.

\subsection{Some Theory for Markov chains}

We will concentrate on the concept of uniform ergodicity (see
Roberts \& Rosenthal (2004) and the references therein). Our
objective is to utilize the following property of uniformly
ergodic Markov kernels $K$; the total variation distance between the
$n-$step kernel and its stationary measure $\pi$,
($\|K^n(x,\cdot)-\pi(\cdot)\|_{TV}$), is bounded above by:
\begin{eqnarray*}
R & = & \big(1-\epsilon\big)^{\left \lfloor \frac{n}{n_{0}}\right\rfloor }
\end{eqnarray*}
where $n_{0}$ and $\epsilon$ are parameters in the following
minorization condition (for $C=E$): $\exists \epsilon>0$ such that
for $\nu$ some non-trivial probability measure and integer
$n_0>0$ and
$\forall x\in C$, $A\in\mathcal{E}$:
\begin{eqnarray}
K^{n_0}(x,A) & \geq & \epsilon\nu(A).
\label{minor}
\end{eqnarray}
If $C\in\mathcal{E}$ is such that (\ref{minor}) is satisfied,
then we term it $(n_0,\epsilon,\nu)-$small.

It follows that, to compare convergence speed (note that
spectral gap techniques - see, for example, Diaconis \& Saloff
Coste (1993) - can also be considered; see Section
\ref{refereeiswrong}) of two uniformly ergodic Markov kernels,
we can compute the pair $(n_{0},\epsilon)$; if one of the
kernels has a substantially larger $\epsilon$ (and smaller $n_0$), then we might
expect it to converge more quickly to the target of interest.
% Prior to this, however, for any proposed algorithm, it is
% necessary to demonstrate that uniform ergodicity holds.

\subsection{Main Result}
\label{mainres}

We now demonstrate that population MCMC approaches can be
defined, using specially constructed component kernels, 
so that they are uniformly ergodic.  For the following Theorem we
denote a mutation (Markov transition) kernel as $K_{M}$ and an
exchange kernel $K_{E}$. A {\em mutation} move is a
Metropolis-Hastings (MH) kernel which attempts to change a
single member of the population, dependent only on the current
state of the chain for that member. An {\em exchange} move is a
Metropolis-Hastings kernel which proposes to swap the current
states of two different members of the population,
mathematically it is defined as,
$K_E^{(il)}:E^2\times\mathcal{E}^{2}\rightarrow[0,1]$:
\begin{eqnarray*}
K_{E}^{(il)}(x_{i,l},dx_{i,l}') & = & \min\bigg\{1,\frac{\pi_i(x_l)\pi_l(x_i)}{\pi_i(x_i)\pi_l(x_l)}\bigg\}
\delta_{x_l}(dx_i')\delta_{x_i}(dx_l') + {}\\
{} & &\bigg[1-\min\bigg\{1,\frac{\pi_i(x_l)\pi_l(x_i)}{\pi_i(x_i)\pi_l(x_l)}\bigg\}\bigg]
\delta_{x_i}(dx_i')\delta_{x_l}(dx_l').
\end{eqnarray*}
For simplicity, we assume that $E_i=E$ $\forall
i\in\mathbb{T}_N$. We now give our main theoretical result;
\vspace{0.1 in}

\begin{res}
\label{res1} Let $\widetilde{K}$ be a Markov kernel defined on a
measurable space $(E^N,\mathcal{E}^N)$ such that:
\begin{eqnarray*} \widetilde{K}(x_{1:N},dx_{1:N}') & = &
\bigg(\tau K_{M} + (1-\tau)K_{E}\bigg)(x_{1:N},dx_{1:N}')\\
K_{M}(x_{1:N},dx_{1:N}') & = & \prod_{i=1}^{N}K_{i}(x_{i},dx_{i}')\\
K_{E}(x_{1:N},dx_{1:N}') & = & \sum_{i=1}^{N-1}\sum_{l=i+1}^{N}\varepsilon_{il}
K_{E}^{(il)}(x_{i,l},dx_{i,l}')
\delta_{x_{-(i,l)}}(dx_{-(i,l)})
\end{eqnarray*}
with
\[
\tau,\varepsilon_{il}\in(0,1) \qquad
\sum_{i=1}^{N-1}\sum_{l=i+1}^{N} \varepsilon_{il}=1
\]
and for each $i\in\mathbb{T}_N$ $K_{i}$ is an aperiodic,
$\lambda$-irreducible Markov kernel with invariant
measure $\pi_{i}(x)\lambda(dx)$. Suppose that $K_{j^{*}}$ is
uniformly ergodic for one $j^{*}\in\mathbb{T}_N$ and for each
$i\neq j^{*}$ $\exists\varrho_{i}\in(0,\infty)$  such that
$\pi_{i}(x)\leq\varrho_{i}\pi_{j^{*}}(x)$ $\forall \; x\in E$. Then
$\widetilde{K}$ is uniformly ergodic.
\end{res}

\begin{proof}See Appendix 2 for the proof.\end{proof}

\noindent\emph{Remark} 1. The assumption of uniform ergodicity
for at least one $K_{j^{*}}$ is not overly restrictive. In many
applications, for example, Bayesian analyses with a proper prior
and a bounded likelihood, a proposal under an independence
kernel with proposal density $q$, where $\pi(x)/q(x)\leq \varrho
~\forall \; x\in E$ (see Tierney (1994), Mengersen \& Tweedie
(1996)), can be found - for example, where $q$ is the prior
density - which ensures uniform ergodicity at the cost of being
a poor proposal for $\pi$. However, if applied to a related
distribution in the population, this proposal may perform quite
well. The assumption $\pi_{i}(x)\leq
\varrho_{i}\pi_{j^{*}}(x)~\forall \; x\in E$ is quite reasonable
and would apply in the framework described here, when $\pi$ is
bounded, for the case $j^{*}=N$. Crucially, for small $N$, we
can establish that the fast rate of convergence for one of the
kernels is propagated through the population (see Section
\ref{convrates}). We note that the result holds for the case
where $K_M$ is a mixture of
kernels, but have omitted the details for brevity.\\

\noindent\emph{Remark} 2. The Theorem shows, in simple cases,
where we can design uniformly ergodic chains for the target
$\pi$ (which are likely to perform well in practice), how to
compare population and single chain approaches. In other words
we can \emph{construct} a population sampler which has a faster
rate of convergence to $\pi^{*}$ (and hence $\pi$) which
justifies the increased cost in computation. Additionally, we
can investigate the population kernels, such as $\widetilde{K}$,
which are likely to provide good mixing for more complex
examples. It should be noted that, due to the limitations of
investigating the minorization conditions, it is difficult to
use this result
when $N$ is large; we discuss this further in Section \ref{refereeiswrong}.\\

\subsection{Impact of the Exchange kernel}

In the following result we show that the exchange kernel, which is
reducible, can indeed improve the rate of convergence relative to a parallel
MCMC algorithm (no interaction between the chains).  To study
this phenomenon, we introduce the following mixing condition (M)
(e.g.~Del Moral (2004)) for Markov kernel
$K:E^{N}\times\mathcal{E}^N\rightarrow[0,1]$: $\exists \;
\epsilon >0$ such that $\forall \; (x,y)\in E^{2N}$
\begin{equation}
\label{M} K(x,\cdot) \geq   \epsilon K(y,\cdot). \tag{M}
\end{equation}
Also, introduce the set $\mathbb{T}_{N}^c
=\{(i,l):i,l\in\mathbb{T}_N, i\neq l\}$.

\begin{prop}
\label{prop1}
Assume $K_M$ satisfies (\ref{M}). Then for any initial distribution $\eta$ and $n\geq 1$:
\begin{eqnarray*}
\|\eta(K_M\circ K_E)^n - \pi^*\|_{TV} & \leq & [2(1-\alpha)(1-\epsilon)]^n||\eta - \pi^*||_{TV}
\end{eqnarray*}
where $\alpha = \sum_{(i,l)\in\mathbb{T}_{N}^c}\varepsilon_{il}[1-\inf_{(x_i,x_l)\in
E^2}\rho_{i,l}(x_i,x_l)]$ and
$\rho_{i,l}(x_i,x_l)  =  \min\bigg\{1,\frac{\pi_i(x_l)\pi_l(x_i)}{\pi_i(x_i)\pi_l(x_l)}\bigg\}$.
\end{prop}

\begin{proof}See Appendix 2 for the proof.\end{proof}

\noindent\emph{Remark} 1. The result provides a sufficient
condition to improve upon parallel MCMC, in that it gives a
tighter bound on the total variation (which is, under (M),
$(1-\epsilon)^{\frac{3n}{2}}||\eta - \pi^*||_{TV}$ - the factor
3/2 is used to make a fair comparison, in terms of CPU time -
for Markov kernel $K_M$): the exchange probabilities are
lower-bounded by $1-\frac{(1-\epsilon)^{\frac{1}{2}}}{2}$.
Essentially, it is crucial to select a good sequence of densities
and use population moves to improve convergence. As noted by a
referee, this indicates that if $\epsilon$ is small, that is
$K_M$ is poorly mixing, the exchange step can help to improve
the convergence rate. In addition, if $\forall (x,y)\in E^2$,
$(i,l)\in\mathbb{T}_{N}^c $ we have $\varrho_L\leq
\frac{\pi_i(x)}{\pi_l(y)}\leq \varrho_U$ then the lower bound
may be achieved if $\frac{\varrho_L}{\varrho_U}\geq
1-\frac{(1-\epsilon)^{\frac{1}{2}}}{2}$, which suggests that we
would require the densities to be similar, to improve
convergence
rates.\\

\noindent\emph{Remark} 2. The result is only of use on compact
spaces (such as for the finite state space variable selection
example of section \ref{BVS} below), where we are able to lower
bound acceptance probabilities and apply the mixing condition
(M). However, even on compact spaces, (M) may be difficult to
verify for MH kernels due to the rejection probability. We note
that this can be circumvented by iterating the kernel.

\section{Example: Bayesian Variable Selection}
\label{BVS} We have established the potential theoretical benefits that
population MCMC methods offer.  A 
specific example is now studied.

\subsection{Model and Data}

Consider the statistical model:
\begin{eqnarray*}
y_{i} & = & \gamma_{0} + \sum_{j=1}^{k_{\textrm{max}}}\vartheta_{j}\gamma_{j}x_{ij}
+ \varpi_{i}
\end{eqnarray*}
with $\varpi_{i}$ i.i.d $\mathcal{N}(0,\sigma^{2})$,
$\vartheta_{j}\in\{0,1\}$, and $\gamma_{j}\in\mathbb{R}$. If we
consider the conjugate prior specification:
$p(\gamma_{0:k}|\sigma,k)=\mathcal{N}_{k+1}(m,\sigma^{2}V)$,
$p(\sigma^{2})=\mathcal{IG}(a,b)$ (where $\mathcal{IG}(\cdot,\cdot)$
is the inverse Gamma distribution) and
\begin{eqnarray*}
p(\vartheta_{1:k_{\textrm{max}}},k)=\frac{1}{k_{\textrm{max}}+1}{k_{\textrm{max}}\choose
k}^{-1}\mathbb{I}_{S_k}(\vartheta_{1:k_{\textrm{max}}})\quad k=0,\dots,k_{\textrm{max}}
\end{eqnarray*}
(where $S_k=\{\vartheta_{1:k_{\textrm{max}}}\in\{0,1\}^{k_{\textrm{max}}}:
\sum_{j=0}^{k_{\textrm{max}}}\vartheta_j=k\}$)
then we can integrate out the parameters $(\sigma,\gamma_{0:k})$ and
sample from a distribution on a finite state space.

We generated 100 data points from a linear model, with $k_{\textrm{max}}=8$
(i.e.~$256$ states). The
posterior probability of a null model was 0.55 and 0.33 for the saturated
model. This is a typical (but simplified) situation for which a standard MCMC sampler would
fail to move around the state space easily.

\subsection{Comparison of Convergence Rates}
\label{convrates}
To sample from the posterior distribution we use an MCMC algorithm
detailed in Denison et al.~(2002) page 53 (with modification
to the variable selection case).

Since the state space is finite, it is clear that an ergodic
Markov chain with appropriate stationary distribution is
uniformly ergodic. To construct an appropriate $\nu$ in
(\ref{minor}) (i.e.~that leads to a large $\epsilon$) we used the
approach discussed in Chapter 6 of Robert \& Casella (2004) (for
example). Let $K_{ij}^{n}$ ($i,j\in E$) denote the $n$ step
transition probability and suppose that
$\displaystyle{\inf_{i}}\{K_{ij}^{n}\}>0$ for some $j$, then
$\forall \; j$
\begin{eqnarray*}
K_{ij}^{n} & \geq & \inf_{l}\{K_{lj}^{n}\} = \epsilon\upsilon_{j}
\end{eqnarray*}
say, where \[\upsilon_{j} =
\frac{\displaystyle{\inf_{i}}\{K_{ij}^{n}\}}{
\displaystyle{\sum_{l\in E}}\inf_{i}\{K_{il}^{n}\}}
~~~\text{and} ~~~\epsilon = \sum_{l\in
E}\inf_{i}\{K_{il}^{n}\}.\] For the algorithm discussed above,
we found that the bound on the rate of convergence was
reasonably similar for $n_{0}=1000$ to $n_{0}=5000$; we focus
upon the pair $(1000,3.63\times 10^{-3})$. To make the analysis
computationally comparable to the population algorithm described
below, we let 50 applications of this kernel be equivalent to a
single step (i.e.~this new kernel has $(20,3.63\times 10^{-3})$
as the $(n_{0},\epsilon)$ pair).

For a population sampler, suppose we take a single auxiliary distribution:
\begin{eqnarray*}
\pi_{2}(\vartheta_{1:k_{\textrm{max}}},k|x_{1:n},y_{1:n})
& \propto & L(x_{1:n},y_{1:n};\vartheta_{1:k_{\textrm{max}}})^{\zeta}
p(\vartheta_{1:k_{\textrm{max}}},k)
\end{eqnarray*}
with $\zeta=0.01$ and $L$ the likelihood function.  
The choice of $0.01$ is used for illustration, and is adopted
to demonstrate the impact of the usage of a related, but easier to sample, distribution in the population.
We
concentrate upon a kernel which updates both chains via the MCMC
algorithm mentioned above for 10 sweeps, followed by an exchange, then another
10 sweeps (which corresponds to the kernel we sample from, that
is, it is a single time step).

It can be shown that the $(n_{0},\epsilon)$ pair for the population
sampler is $(1,6.01\times 10^{-4})$. This was computed by finding
the $\epsilon$ and $\nu$ in the minorization condition for $\pi_2$ (as above)
and then:
\begin{eqnarray*}
\phi & = & \sum_{(\vartheta_{1:k_{\textrm{max}}},k)\in E}\nu(\vartheta_{1:k_{\textrm{max}}},k)
\min\bigg\{1,\frac{\pi_1(\vartheta_{1:k_{\textrm{max}}},k)}{\pi_2(\vartheta_{1:k_{\textrm{max}}},k)\varrho_1}\bigg\}\\
\varrho_1 & = & \sup_{(\vartheta_{1:k_{\textrm{max}}},k)\in E}\frac{\pi_1(\vartheta_{1:k_{\textrm{max}}},k)}{\pi_2(\vartheta_{1:k_{\textrm{max}}},k)}
\end{eqnarray*}
the constant in the minorization condition is then $\epsilon^2\phi$.

The bound on total variation distance, $R$, suggests a much faster rate of convergence for the
population algorithm: $M_{0.01}=25326$ (the
number of iterations to achieve a bound on the total variation
distance less than 0.01) for the vanilla algorithm and
$M_{0.01}=7660$ for the population algorithm, that is,
it is significantly faster.

This example demonstrates that a simple extension of the original
algorithm to include an auxiliary distribution that provides a good
proposal (for the original target) allows efficient movement around
the state space. That is, it reiterates the point in remark 1 to Theorem \ref{res1},
that the fast convergence rate of one of the chains is propagated through the system.

\subsection{Summary}
\label{refereeiswrong} A point of interest, raised by a referee,
is that we found as we increased $N$ we
found it difficult to find a Markov kernel that could improve
upon rate of convergence to stationarity when $N=2$. 
This can be explained as follows;
suppose, for illustration, we run $N$ parallel chains with the
same (marginal) invariant measure $\pi(x)\lambda(dx)$ with
kernel $K_M\circ K_E \circ K_M$ in the proof of Theorem
\ref{res1}. Assume that all the mutation kernels, except one,
are strongly aperiodic $(n_0=1)$ and mix almost perfectly (that
is, are uniformly ergodic with $\epsilon\approx 1$). Then we can
easily establish that the constant in the minorization condition
is $\epsilon^{N} \rightarrow 0$ as $N\rightarrow \infty$, so
that $R\approx 1$. This establishes two points:

\begin{itemize}

\item Firstly, that continually extending the state-space will
inevitably lead to slower convergence of the Markov chain,
unless very efficient population moves may be constructed.  That
is, we should not naively extend the state-space and expect our
Markov chain to converge more quickly; unlike convergence of
particle algorithms (Del Moral, 2004), there is not necessarily
an improved convergence property as $N\rightarrow\infty$.
This is an open problem for future research.

\item Secondly, that in cases where $R$ is not informative on
the rate of convergence, it may be more beneficial to investigate
other properties of the Markov kernel, such as the spectral gap.
Such an analysis is likely to be
far more involved than discussed here; see Madras \& Zheng (2003)
for example.

\end{itemize}

\section{Population Moves for the Mixture Example}

Now that we have established, for difficult problems, that
population methods can lead to faster convergence, we discuss how to
implement population moves for our mixture example (Section
\ref{rjex}). Our notation is such that $\theta_{i}=(\eta^{i},w^{i},
\Psi^{i},k_{i})$, $i\in\mathbb{T}_N$ and $(\eta^{i},w^{i})$ refers to
all of the component specific parameters and weights for chain $i$.
In our tempering approach (i.e.~$\pi^{\zeta_i}$), we will temper the likelihood terms only,
rather than the full posterior, to avoid any integrability problems.
We now proceed to combine several MCMC methods to improve the mixing
ability of the chain.

\subsection{Exchange Moves}
\label{exch} An exchange move is used to swap information
between two different parallel tempered chains.  Our strategy is
as follows: at iteration $t$ we select two adjacent chains (in
terms of the temperature parameter) uniformly at random and
propose to swap their values. In order to achieve a reasonable
interaction between the chains, the temperature ladder is set so
that this move is accepted about half of the time (Liu, 2001); see
Section \ref{genex} for further discussion.

One way to improve this move is to use the delayed rejection method,
as suggested by Green \& Mira (2001). At iteration $t$, we select
any two chains $i_{1}$ and $i_{2}$ to swap, accepting or rejecting
with the usual Hastings ratio, that is, with probability
\begin{eqnarray*}
\rho_{1}(\theta_{1:N},\theta_{1:N}') & = & \min\Bigg\{1,\frac{\pi_{i_{1}}(\theta_{i_{2}})
\pi_{i_{2}}(\theta_{i_{1}})}{\pi_{i_{1}}(\theta_{i_{1}})\pi_{i_{2}}
(\theta_{i_{2}})}\Bigg\}
\end{eqnarray*}
where the labelling of the chains is with respect to the current
state of the chain and $\theta_{1:N}'$ denotes the new configuration
of chains.  If this is rejected, we select two adjacent chains
$i_{3}$ and $i_{4}$ to swap, denoting this configuration
$\theta_{1:N}''$.  To ensure reversibility with respect to the
target, as part of the delayed rejection method, we construct a
pseudo move which consists in starting from the second stage
proposed state $\theta_{1:N}''$, proposing to move to
$\theta_{1:N}^*$ (which swaps $\theta_{i_1}''$ and $\theta_{i_2}''$)
and rejecting it with probability
$\rho_1(\theta_{1:N}'',\theta_{1:N}^*)$. The second stage move is
accepted with probability
\begin{eqnarray*}
\rho_{2}(\theta_{1:N},\theta_{1:N}'') & = & \min\Bigg\{1,\frac{\pi_{i_{3}}(\theta_{i_{4}})
\pi_{i_{4}}(\theta_{i_{3}})\{1-\rho_{1}(\theta'',\theta^{*})\}}
{\pi_{i_{3}}(\theta_{i_{3}})\pi_{i_{4}}(\theta_{i_{4}})
\{1-\rho_{1}(\theta,\theta')\}}\Bigg\}.
\end{eqnarray*}
This move allows for increased interaction within the population.
At the first stage, we allow any pair of chains to be swapped, thus
if a state of a chain at a high temperature is consistent with one of the
distributions at a lower temperature, it is allowed to quickly jump
down the ladder.

\subsection{Crossover Moves}

Liang \& Wong (2001) employ various crossover moves in an
evolutionary MC algorithm; the objective is to increase the interaction
within the population.  In our algorithm, we use two move types:

\subsubsection*{I. Variable dimension crossover}
To construct a move likely to have high acceptance probability
in the mixture model, we begin by reassigning mixture component
labels in each chain so they satisfy an ordering constraint on
the weights, in order approximately to match the labels of
components in different chains. When in state $\theta$, we
select a variable dimension crossover with probability
$v(\theta)$;
\begin{displaymath}
v(\theta) = \left\{\begin{array}{ll}
1 & \textrm{if}~k_{i}\neq k_{j}~\textrm{for some}~i\neq j \\
0 &  \textrm{otherwise.}
\end{array}\right.
\end{displaymath}
Note the case $v(\theta)=0$ corresponds to a `do nothing' move. We
select a pair of chains with differing dimension with probability
inversely proportional to the squared difference between the
dimensions. We then propose the new state of the population members,
by swapping $k$'s and the weights.  We take the lowest weighted
component specific parameters of the higher dimensional chain to the
lower dimensional chain, i.e. if $k_{i_{1}}>k_{i_{2}}$ for the
selected chains $i_{1},i_{2}$ we propose:
\begin{eqnarray*}
\eta_{i_{1}}' & = & (\eta_{k_{i_{1}}-k_{i_{2}}}^{i_{1}} ,\dots, \eta_{k_{i_{1}}}^{i_{1}})\\
\eta_{i_{2}}' & = & (\eta_{1}^{i_{1}},\dots, \eta_{{k_{i_{1}}-
k_{i_{2}-1}}}^{i_{1}}, \eta_{1}^{i_{2}},
\dots,\eta_{k_{i_{2}}}^{i_{2}})
\end{eqnarray*}
where $\eta^{i_{1}}_{j}$ denotes an element of $\eta^{i_{1}}$.
The acceptance probability is easily calculated and thus
omitted. After the move has been accepted or rejected (or the do
nothing move) we propose (and accept) a random permutation (all
permutations have uniform probability of being proposed) of the
labels of the parameters (of all the chains).  This final
permutation ensures invariance with respect to the
target.

% ; it is a step which is necessary and legitimate in the
% mixture model, where the model itself is invariant to
% permutation of the component labels.  In general, crossover
% moves that exchange subsets of parameters between chains would
% not require the permutation step.

\subsubsection*{II. Fixed Dimension Crossover}

We begin by reassigning the mixture component labels in each
chain by ordering on the first dimension of the means. When in
state $\theta$, we select a fixed dimension crossover with
probability 1, if it can be selected (i.e.~there are at least
two chains with the same dimensionality) otherwise we select a
`do nothing' move. Select a pair of chains $(i_{1},i_{2})$ with
the same dimensionality, with probability
\begin{eqnarray}
p(i_{1},i_{2}|\theta) & \propto & |\zeta_{i_{1}}-\zeta_{i_{2}}|^{-1}\mathbb{I}_{k_{i_{1}}=k_{i_{2}}}.
\label{selec}
\end{eqnarray}
We select a position $j=1,\dots,k_{i_{1}}-1$ to crossover,
this selection made with probability proportional to $1/j$ and
switch all component specific parameters to the left of $j$
inclusive (note that if the identifiability constraint is not
satisfied in the proposed state of the chain we immediately reject)
and accept/reject on the basis of the Hastings ratio. After the accept/reject
(or the do nothing move) decision has been made, we again propose a random permutation of the
labels of the parameters.

\subsection{Snooker Jumps}
One of the most important ways we can use the information in the
population is by targeting variable dimension jumps by using another
chain. This idea is linked to the snooker algorithm of Gilks et
al.~(1994) and is performed in the following way.  When in state
$\theta$, we select a birth with probability $b(\theta)$, where
\begin{displaymath}
b(\theta) = \left\{\begin{array}{ll}
1 & \textrm{if}~k_{i}=1~\forall \; i \\
0 & \textrm{if}~k_{i}= k_{\textrm{max}}~\forall \; i \\
1/2 & \textrm{otherwise}
\end{array}\right.
\end{displaymath}
then select a chain (the current point $\theta_{c}$) for which a
birth is possible (let $m_{b}(\theta)$ be the number of chains such
that a birth can occur when in state $\theta$) with uniform
probability, and select an anchor point ($\theta_{a}$) with
probability inversely proportional to the absolute value of the
difference between the inverse temperatures. We then generate
$w\sim\mathcal{B}e(1,k_{c})$, with $\mathcal{B}e(\cdot,\cdot)$ the
beta density, and draw a new $\mu ,\Lambda$ pair from:
\begin{eqnarray*}
q(\mu,\Lambda) & = & \sum_{j=1}^{k_{a}}\bar{h}(\eta_{j})
\mathcal{N}_{r}(\mu_{j}^{a},\sigma)\mathcal{IW}(2r+3,\Lambda_{j}^{a})
\end{eqnarray*}
where
\begin{eqnarray*}
\bar{h}(\eta_{j}) & \propto & \frac{1}{k_{c}}\sum_{l=1}^{k_{c}}h(
(\mu_{j}^{a},\Lambda_{j}^{a}),
(\mu_{l}^{c},\Lambda_{l}^{c}))
\end{eqnarray*}
and $h(\cdot,\cdot)$ is the Mahalanobis distance. We then perform
the rest of the move as for the birth in Appendix 1. In the death
move, we perform much the same as for Appendix 1, except we select a
current point with probability $1/m_{d}(\theta)$ (where
$m_{d}(\theta)$ is the number of chains for which a death can occur
when in state $\theta$) and (redundantly) select an anchor point
(which is used in the reverse birth). The birth move is accepted
with probability $\min\{1,A\}$ with:
\begin{eqnarray*}
A & = & \frac{p(y_{1:n}|\eta^{'c},
w_{1:k_c}^{'c},k_{c}+1)^{\zeta_{c}}p(\mu)p(\Phi)p(k_{c}+1)}
{p(y_{1:n}|\eta^{c},w_{1:k_c-1}^{c},k_{c})^{\zeta_{c}}p(k_{c})}
B(k_{c}\delta,\delta)^{-1}w^{\delta-1}(1-w)^
{k_{c}(\delta-1)}\frac{(k_{c}+1)!}{k_{c}!}{}\nonumber \\
& & {}\times
\frac{d(\theta')m_{b}(\theta)}{(k_{c}+1)b(\theta)m_{d}(\theta')}
\frac{(1-w)^{k_{c}-1}}{\mathcal{B}e(w;1,k_{c}) q(\mu,\Phi)}
\end{eqnarray*}
where $\Phi$ is the Cholesky decomposition of $\Lambda$ (see
Appendix 1 for details),
$p(y_{1:n}|\eta^{c},w_{1:k_c-1}^{c},k_{c})^{\zeta_{c}}$ is the
tempered likelihood (for the current point), $B(\cdot,\cdot)$ is the
beta function and $\mathcal{B}e(x;\cdot,\cdot)$ is the beta density
evaluated at $x$. The objective of this move is to propose new
component-specific parameters which are likely to be consistent with
the data, but are markedly different from the current components. It
also provides an adaptive element to the birth proposal, as it
relies on current information, in the population, that is being
continuously updated.

\subsection{Constrained Sampling}

One aspect of population-based simulation that is apparent is the
need to maintain diversity of the population (as in sequential Monte
Carlo - see Del Moral et.~al.~(2006)). In many cases for which it is
difficult to traverse the state space, it is often the case that the
chains at lower temperatures can become trapped (stuck in
local modes) as for single chain MCMC methods, and this may lead to inaccurate Monte Carlo estimates
of quantities of interest.
To avoid this problem, we propose to
constrain some of the members of the population, that is, for some
subset $\mathbb{T}_{l:N}$ ($l\geq 2$), and $i\in \mathbb{T}_{l:N}$, $\pi_{i}$ is a
density constrained to $E_{i}\subset E$.  In the setting of
trans-dimensional problems, a natural choice of sets $E_{i}$ may be
selected with respect to the model dimension
(e.g.~$E_{i}=\bigcup_{k\in\mathcal{K}_{i}} \big(\{k\}\times
E^{k}\big)$ for $\mathcal{K}_{i}\subset\mathcal{K}$ and $i\in \mathbb{T}_{l:N}$).
In general, choosing an appropriate $E_i$ is challenging; see
Atchad\'e \& Liu (2004) (unpublished technical report) for some discussion.
We remark that, in the example in Section \ref{genex},
this technique will prove to be very important.

\subsection{The Algorithm}
\label{alg}
To sample from the augmented distribution we use the algorithm
below; we use the genetic algorithm terminology of Liang \& Wong
(2001).
\begin{flushleft}
0. \emph{Initialise the chain $\theta$}.\\
\emph{For $t=1,\dots,M$ sweep over the following:}\\
1. \emph{Mutation. Select a chain $i\in\{1,\dots,N\}$ with probability $\tau_{i}$ and then
perform one sweep of the reversible jump algorithm in Appendix 1
for this chain (with appropriate modification for constrained targets).}\\
2. \emph{Make a random choice between performing steps 3 or 4.}\\
3. \emph{Crossover. Propose a variable dimension crossover move with
probability $1/2$, else propose a fixed dimension crossover.}\\
4. \emph{Snooker Jump. Propose a birth with probability $b(\boldsymbol{\theta})$,
else propose a death.}\\
5. \emph{Exchange. Perform the delayed rejection exchange move.}
\end{flushleft}
Note that, for constrained chains, we only allow them to be involved
in fixed dimensional crossovers and a special exchange move that
augments step 5; we propose to exchange a constrained and
non-constrained chain, selecting the move only if such a move may be
performed (that is $k_i\in\mathcal{K}_j$ and $k_j\in\mathcal{K}_i$
for $i\neq j$). All selections are made with uniform probability and
no delayed rejection is used.

\section{Gene expression example revisited}
\label{genex}
\subsection{Specification of Simulation Parameters}

\noindent\textbf{Population Size :} To run the population algorithm
in Section \ref{alg} we used $N=25$ with 5 chains constrained. We
recommend a large population size in general (although the
discussion in Section \ref{refereeiswrong} should be considered), so
that results are reasonably similar for separate runs of the
algorithm. See Jasra, Stephens \& Holmes (2006) (unpublished technical
report, available from \texttt{http://stats.ma.ic.ac.uk/das01/}) for more guidance.

\vspace{0.1 in} \noindent\textbf{Temperature Parameters:} For the
main population (the unconstrained chains) the following inverse
temperatures were selected:
\begin{eqnarray*}
\zeta_{1} & = & 1\\
\zeta_{i} & = & \zeta_{i-1} - \varsigma\varphi^{i-1}~i\in\mathbb{T}_{2:20}
\end{eqnarray*}
for constants $\varsigma>0$, $\varphi>1$. We selected
$\varsigma=10^{-6}$, $\varphi=1.85$ ($\zeta_{20}=0.74$). Our choice
of cooling schedule, and population size was based upon pilot
tuning.  
We selected a slowly decreasing sequence of $\zeta$'s since
we observed a poor acceptance rate for the exchange move for
distributions that were further away from each other, as expected.
We found that the inverse temperature at which the reversible jump
algorithm performed best (that is, reasonable acceptance rates along
with regions of high support that were similar to the target) was
$\zeta=0.75$; thus we attempted to include a distribution with this
temperature. We note that we need to be careful when specifying
temperatures, since for low temperatures (low depending on the
problem at hand) the distribution starts to favour dimensionalities
that are small, although this may be alleviated by specifying priors
for $k$ which penalise small values. For this particular problem, we spent
at least 2-3 hours in tuning the parameters; a more automatic procedure may
be implemented - see below.

The term `reasonable acceptance rate' deserves a
quantitative definition. 
A useful criterion put forward by Iba (2001) is that
the expected (wrt target) log Hastings ratio of the exchange probability is equal to 1. This in turn means that, approximately, the algorithm will accept the exchange about half the time. 
We remark that such an approach may be used to provide an automatic
temperature selection and
will vastly
reduce the time spent on tuning the algorithm (as noted by a referee
this aspect should be taken into account in the computational cost and
comparison of the algorithm). We refer the reader to Iba (2001) for the details
and to Goswami \& Liu (2006) (unpublished technical report) for an alternative approach.

\vspace{0.1 in} \noindent\textbf{Constrained Targets:} To select the
$E_i$, $i\in\mathbb{T}_{21:25}$ we used a pilot simulation. We ran the
algorithm with $N=25$ and the inverse temperature parameters
discussed above, only 9 chains in the main population and 16
constrained chains (given inverse temperature parameter 0.999). We
selected the subspaces $E_i$ with respect to the dimensionality,
that is we had 10 chains constrained to lie
$k\in\mathbb{T}_{1:2},\dots,\mathbb{T}_{19:20}$ then six other chains constrained to
lie in $k\in\mathbb{T}_{3:6}$, $\mathbb{T}_{6:9}$, $\mathbb{T}_{9:12}$,
$\mathbb{T}_{12:15}$ and $\mathbb{T}_{15:18}$. This was adopted in order
to determine whether there was any support outside $\mathbb{T}_{3:11}$
found in Section \ref{rjex}. Based upon a short run, the
five constrained chains were taken as $\mathbb{T}_{2:4}$, $\mathbb{T}_{4:6}$,
$\mathbb{T}_{5:7}$ and $\mathbb{T}_{7:9}$, $\mathbb{T}_{9:11}$. The idea of the pilot
tuning is to avoid wasting CPU time on population members
constrained to lie in areas of support that have low density with
respect to the original target density. Additionally, the
constrained chains need to be able to interact with the main
population and the preliminary tuning allows us to make this choice. The
inverse temperature parameters for the constrained chains were
0.999, since we seek to maintain diversity with respect to the
population at colder temperatures. For further discussion in the
setting of partitions, especially in the context of reversible jump,
see Atchad\'e \& Liu (2004) (unpublished technical report).

The algorithm in Section \ref{alg} was run for 1 million sweeps which took approximately 9$\frac{1}{2}$ hours (code is available on request from the first author). A minor modification to algorithm was made:
that if a crossover move was selected, we also propose an exchange for the
chain of interest.

\subsection{Comparison with Vanilla Sampler}

The improvement over the vanilla sampler is
substantial, on average, the chain of interest took 6.75 sweeps to jump
between a mixture model with less than 4 components to a mixture with
more than 6.

The observed inability of the vanilla reversible jump algorithm to move
around the state space (from $k\in\mathbb{T}_{3:5}$ to
$k\in\mathbb{T}_{8:11}$) is not present for the population sampler,
since it may represent both parts of the space simultaneously.
Note that,
due to the complexity of the target, we cannot claim that the sampler has
converged; there may be regions of high posterior density that
are still unexplored.

It is apparent, from simulations, that despite the substantially
improved performance when compared to the vanilla sampler, the sampler has
missed the region $k\in\{8,9,10\}$. This is due to the reasons
discussed earlier, and demonstrates that constrained chains can help
guard against such problems.

The effective sample size (ESS) is a standard measure of the
relative efficiency of an MCMC sampler (see Liu (2001) for full
definition).  For our population algorithm, the ESS for $k$ was
59745 (60000 samples, using a lag of 10 in the autocorrelation
calculation), compared with 2998 and 3591 for both vanilla
algorithms. Taking
\begin{eqnarray*}
\mathsf{E} & = & 2\frac{\textrm{ESS}_{\textrm{pop}}}{M_{\textrm{pop}}T_{\textrm{pop}}}/
\bigg[\frac{\textrm{ESS}_{\textrm{van}_{1}}}{M_{\textrm{van}_{1}}T_{\textrm{van}_{1}}}
+\frac{\textrm{ESS}_{\textrm{van}_{2}}}{M_{\textrm{van}_{2}}T_{\textrm{van}_{2}}}\bigg]
\end{eqnarray*}
where the subscripts refer to the population, vanilla algorithm runs 1 and
2 (the chains run in Section \ref{rjex}) and chains respectively, $M$ is
the sample size and $T$ is the CPU time, we obtain $\mathsf{E}=2.07$
(we note that similar conclusions are drawn when taking into account tuning
time).
Therefore there is little contest between using population-based
reversible jump and the vanilla counterpart for this example; the
population approach is far superior (note that all coincidental
simulation parameters are the same between algorithms).

\subsection{The Efficiency of Sampler Moves}
The exchange move was accepted 44$\%$ of the time at the first stage
and 75$\%$ at the second. This indicates that delayed rejection
helps to ensure that the algorithm is constantly swapping information between
the chains; for 86$\%$ of the sweeps there is at least one exchange.

The snooker and variable dimension crossover moves have acceptance
rate less than $1\%$. That this occurs is to be expected. Liang \&
Wong (2001) report fairly small acceptance rates for their crossover
moves; our rates are smaller as our algorithm operates upon are a more complex
space. Our experience with the snooker and
variable dimension crossover in more
simple examples, is that they are generally not worth the extra
coding effort given their performance. However, we were satisfied
with the fixed dimensional crossover which was accepted $2.9\%$ of
the time. The snooker birth is accepted more often than the standard
birth, but the reverse snooker death move is rarely accepted (c.f.~a
birth move with a proposal that has low variance). Hence the move is
less successful overall than the standard birth. The variable
dimension move acceptance rates (averaged over all chains) were
still below $1\%$, with the split/combine move being less effective.

Overall, our recommendation is that the exchange (with delayed rejection)
and fixed dimensional crossover are implemented. This is in conjunction
with constrained targets.

\subsection{Comparison with Simulated Tempering}

A more appropriate single chain sampler to compare with the
population method is a simulated tempering algorithm; see
Hodgson (1999) for an example of another variable dimension
simulated tempering algorithm. Here the target distribution is:
\begin{eqnarray}
\pi(\theta,\zeta,k|y_{1:n}) & \propto & p(y_{1:n}|\theta,k)^{\zeta}
p(\theta,k)p(\zeta)\label{eq:stpi}
\end{eqnarray}
that is, $\zeta\in Z$ is stochastic
with $Z$ finite and {pseudo prior $p(\zeta)$ (in this
example we set the prior as opposed to constructing one adaptively;
see Geyer \& Thompson (1995)). Note that our formulation does not 
require any normalization constants to be known, although such an
algorithm will allow us to find a `good' pseudo prior; see Geyer \& Thompson
(1995).

To
sample from (\ref{eq:stpi}) we use the reversible jump algorithm
in Appendix 1, conditional upon $\zeta$, and to update $\zeta$ a
delayed rejection move was adopted: propose a temperature
uniformly at random from $Z$ and accept or reject with the Hastings
ratio - if rejected, select an adjacent temperature and perform a
pseudo move that selects to move from the proposed temperature at
the second stage to the proposed temperature at the first stage.

Using some trial simulations we were unable to find a pseudo prior so that for a reasonable number of
temperatures (e.g.~$25$), the algorithm could jump between the target and the inverse temperature 0.75. As a result, for any sensible number of
distributions, the performance of this approach only slightly improves over the vanilla algorithm. An example of a run of the simulated
tempering algorithm was setting $p(\zeta = \zeta_i) \propto1/i$,
with $\zeta_{1}=1$ and having a difference of $1\times 10^{-4}$
between each temperature.  We found, with $|Z|=25$, that the
algorithm only visited the distribution of interest $10\%$ of the
time in a run of 250000 sweeps.

In this Section, it was seen that simulated tempering has been difficult
to set up, so that it can be operated efficiently; see Atachd\'e \& Liu (2004)
(unpublished technical report) for a more automatic procedure.
In addtion, see Zheng (2003) for a theoretical comparison of population MCMC and simulated tempering.

% The results above justify why a significant proportion of the population %had marginal targets similar to $\pi$; they contain information
% relevant to the target, with slightly improved mixing. 
% That the
% distribution with inverse temperature 0.75 is difficult to reach
% under the simulated tempering approach does not matter in the
% population algorithm; we use this distribution (occasionally) to
% discover states with high posterior support, whilst the
% distributions close to the target retain the places the chain has
% visited. 

\section{Discussion}

To summarize, in our experience a vanilla reversible jump algorithm
often fails to explore the support of a multimodal model space. We
introduced population-based reversible jump MCMC and gave some
theoretical justification for why these methods can be preferable to
standard MCMC methods. In addition, it was demonstrated that population-based reversible jump
is a means to improving variable dimension simulation. Overall, our
method helps open up the possibility of fully Bayesian analyses in
problems for which simulation is prohibitively slow.  Note that
the basic algorithm (without population moves) can easily be coded
given a vanilla sampler. Therefore population MCMC provides a simple way
to check the performance of MCMC algorithms.

One of the problems of our approach is the limited amount of success
of our crossover moves. Whilst this was observed for
simpler problems in Liang \& Wong (2001), we would still hope that
the population can provide more information when proposing moves. 
% We
% note, however, that we do not want too many snooker type moves (see
% Gilks et al.~(1994)) to be accepted. This is because it will reduce
% the diversity of the population. 
General guidelines for constructing
constrained subspaces of $E$ and finding efficient ways to make them
interact with the population is an important area for further
investigation. There are many potential methodological extensions
that may be considered.

Firstly, to combine our approach with adaptive MCMC methods (see Andrieu \& Moulines (2006) and Chauveau \&
Vandekerkhove (2002) in the population context). This is likely to
be superior to standard adaptive algorithms, since there is more
information to update proposals. Furthermore, there is more
information in terms of where the chain has not been, i.e.~we
may search (fewer) regions of the support of $\pi$ for states with
high density.

Secondly, we may consider combining some of our ideas with
many recent stochastic simulation algorithms. For example, we might
use the constrained chains in a similar context to the equi-energy sampler
of Kou et al.~(2006). In the trans-dimensional case, the energy rings 
could be replaced
with dimension rings. This
is likely to produce a highly diverse sample with respect to the
dimensionality.

\section*{Acknowledgement}
The first author was supported by an PhD EPSRC Studentship. We thank
two referees and an associate editor for comments which lead to a substantial improvement
in the content of the paper.
The first author acknowledges
many useful conversations with Nick Heard, Matthew Gander and Zhicheng Zhang.
We also thank Adam Johansen and Tso-Jung Yen for some comments on earlier versions of the paper.

\subsection*{Appendix 1: Reversible Jump Sampler}

The vanilla sampler used in Section \ref{rjex} is now outlined.

One of the drawbacks of the model we have selected is the need for $\Lambda$
to be positive definite. As a result, moves in MCMC
simulation will be difficult to construct such that this constraint is satisfied.
To deal with this problem
we consider the Cholesky decomposition $\Phi$ (see Dellaportas
\& Papageorgiou (2006) for an analysis using the spectral decomposition).
That is,
$\Lambda = \Phi\Phi'$ where
$\Phi$ is lower triangular with positive diagonal elements
(recall the Jacobian is $2^{r}\prod_{l=1}^{r}\phi_{ll}^{r-l+1}$).
Our RJMCMC algorithm is as follows (all moves are Metropolis-Hastings steps
unless otherwise stated).

Firstly the fixed dimensional moves.
The component specific means ($\mu_{j}$) and component
specific lower triangular part of $\Phi_{j}$ are both updated
via an additive cauchy random walk, independent in each dimension. The component
specific diagonals of $\Phi_{j}$ are updated via a multiplicative log-normal random walk, independent in each dimension. The weights are proposed
using an additive normal random walk on the logit scale. Finally, $\Psi$
is generated using a Gibbs kernel; the full
conditional is $\mathcal{W}(2(g+k\alpha'),(2h+
2\sum_{j=1}^{k}\Lambda_{j}^{-1})^{-1})$.

Secondly a birth/death
of a component, largely following Richardson \& Green (1997). Briefly, we draw a new $\mu$ and $\Phi$
from the prior and $w\sim\mathcal{B}e(1,k)$, setting the new weights as
$(w_{1}(1-w),\dots,w_{k}(1-w),w)$, selecting the move with probability $b_{k}$ (when in state $k$). The death, selected
with probability $d_{k}$, is performed by selecting a component to die with uniform probability and inverting the jump function.

Finally, a split/combine of a component. We select a split with probability $s_{k}$ and
choose a component $j^{*}$ uniformly at random to split into components labelled as $(j_{1},j_{2})$. The split requires the following actions:
\begin{quote}
(i) Split the weight by drawing $u_{1}\sim\mathcal{B}e(\gamma,\gamma)$ and set
\begin{eqnarray*}
w_{j_{1}} & = & u_{1}w_{j^{*}}\\
w_{j_{2}} & = & (1-u_{1})w_{j^{*}}.
\end{eqnarray*}
(ii) Split the mean vector by drawing $u_{1(2)},\dots,u_{r(2)}\sim\mathcal{N}(0,\sigma_{\mu})$ and take
\begin{eqnarray*}
\mu_{l(j_{1})} & = & \mu_{l(j^{*})} + u_{l(2)}\\
\mu_{l(j_{2})} & = & \mu_{l(j^{*})} - u_{l(2)}.
\end{eqnarray*}
(iii) Split the off diagonals of $\Phi$ by drawing $u_{21(3)},\dots,u_{r(r-1)(3)}\sim
\mathcal{N}(0,\sigma_{\phi})$ and take
\begin{eqnarray*}
\phi_{lm(j_{1})} & = & \phi_{lm(j^{*})} + u_{lm(3)}\\
\phi_{lm(j_{2})} & = & \phi_{lm(j^{*})} - u_{lm(3)}
\end{eqnarray*}
where $l=2,\dots,r$, $m=1,\dots,l-1$.\\
(iv) Split the diagonals of $\Phi$ by drawing $u_{11(3)},\dots,u_{rr(3)}\sim
\mathcal{LN}(0,\sigma)$ and take
\begin{eqnarray*}
\phi_{ll(j_{1})}  =  \frac{\phi_{ll(j^{*})}}{u_{ll(3)}} ~~~~~
\phi_{ll(j_{2})}  =  \phi_{ll(j^{*})}u_{ll(3)}.
\end{eqnarray*}
\end{quote}
In order to combine we select the move with probability $c_{k}$ and invert the jump function
above. We note that due to the symmetry constraint imposed on the jump function it does not
matter which way we combine the components (see Capp\'e et al.(2003) for details on this).
We choose two components to combine, when in
state $k$, with probability
inversely proportional to the Mahalanobis distance between them, that is:
\begin{eqnarray*}
p_{k}(j_{1},j_{2}) & \propto & \big[(\mu_{j_{1}}-\mu_{j_{2}})'
\Lambda^{-1}_{j_{1}}
(\mu_{j_{1}}-\mu_{j_{2}})
+(\mu_{j_{2}}-\mu_{j_{1}})'
\Lambda^{-1}_{j_{2}}
(\mu_{j_{2}}-\mu_{j_{1}})\big]^{-1}.
\end{eqnarray*}
The split in state $k$ is accepted with probability $\min\{1,A\}$ where
\begin{eqnarray*}
A & = & (\textrm{likelihood ratio})\frac{p(\Phi_{j_{1}})
p(\Phi_{j_{2}})}
{p(\Phi_{j^{*}})}
\frac{p(\mu_{j_{1}})
p(\mu_{j_{2}})}{p(\mu_{j^{*}})}
B(k\delta,\delta)^{-1}\big(w_{j^{*}}u_{1}(1-u_{1})\big)^{\delta-1}
\times {} \\
& & {} \frac{p(k+1)}{p(k)}\frac{(k+1)!}{k!}
\frac{kc_{k+1}p_{k+1}(j_{1},j_{2})}{s_{k}}\frac{|J|}{2q_{1}(u_{1})
q_{2}(u_{2})q_{3}(u_{3})}
\end{eqnarray*}
where $|J|$ is the Jacobian:
\begin{eqnarray*}
|J| & = & 2^{\frac{r(r+3)}{2}}w_{j^{*}}\prod_{l=1}^{r}\frac{\phi_{ll(j^{*})}}
{u_{ll(3)}}
\end{eqnarray*}
and obvious notation for the prior and proposal densities.

The algorithm is performed in a
deterministic sweep over all fixed dimension moves followed by a
random choice of birth/death or split/merge. The particular trans-dimensional move is selected with uniform
probability (assuming we allow a move, i.e.~no birth or split when
$k=k_{\textrm{max}}$ or death or combine when $k=1$).

\subsection*{Appendix 2: Proofs}
\begin{proof}[Proof of Theorem \ref{res1} ]
Our strategy is to show that $\widehat{K}^{N-1}(x_{1:N},A)=(K_M\circ K_E\circ K_M)^{N-1}(x_{1:N},A)$
($A\in\mathcal{E}^N$)
is a uniformly ergodic Markov kernel. Then to prove $\widetilde{K}^{3(N-1)}$
is uniformly ergodic, we may use the fact that $\widetilde{K}^{3(N-1)}(x_{1:N},A)
\geq \tau^{2(N-1)}(1-\tau)^{N-1}\widehat{K}^{N-1}(x_{1:N},A)$ (proved below).
We begin by proving the case $N=2$ and then use an induction on $N$.
The strategy of the proof is to use the uniform ergodicity of $K_{j^{*}}$
(which we take to be $K_{N}$) and the acceptance
of an exchange move. We assume $n_0=1$ as the proof
can be extended to the case $n_0>1$ with only notational complications.
We denote $x_{i}^{(l)}$ as the value of $x_{i}$ after $l$ steps.

We first establish $\widetilde{K}^{3(N-1)}(x_{1:N},A)
\geq \tau^{2(N-1)}(1-\tau)^{N-1}\widehat{K}^{N-1}(x_{1:N},A)$, $A\in\mathcal{E}^N$.
Consider $\widetilde{K}^{3}$, then:
\begin{eqnarray*}
\widetilde{K}^{3}(x_{1:N},A)
& = & \int_{E^{N}}
\widetilde{K}(x_{1:N},dx^{(1)}_{1:N})\int_{E^N}
\widetilde{K}(x_{1:N}^{(1)},dx^{(2)}_{1:N})
\widetilde{K}(x_{1:N}^{(2)},A)\\
& \geq &
\int_{E^{N}}
\widetilde{K}(x_{1:N},dx^{(1)}_{1:N})\int_{E^N}
\widetilde{K}(x_{1:N}^{(1)},dx^{(2)}_{1:N})\tau K_{M}(x_{1:N}^{(2)},A)
\end{eqnarray*}
where we have applied Chapman-Kolmogorov and used
$\widetilde{K}(x_{1:2},\cdot)\geq \tau K_{M} (x_{1:N},\cdot)$
$\forall \; x_{1:N}\in E^{N}$. We then have
\begin{eqnarray}
\widetilde{K}^{3}(x_{1:N},A) & \geq & \tau^{2}(1-\tau)\int_{E^N}
K_{M}(x_{1:N},dx_{1:N}^{(1)})
\int_{E^{N}}
K_{E}(x_{1:N}^{(1)},dx_{1:N}^{(2)})
K_{M}(x_{1:N}^{(2)},A)\label{peq1}
\end{eqnarray}
which corresponds to selecting a mutation followed by
an exchange and then followed again by another mutation. This clearly
follows for $\widetilde{K}^{3(N-1)}(x_{1:N},A)$, by
the argument above:
\begin{eqnarray*}
\widetilde{K}^{3(N-1)}(x_{1:N},A) & \geq & \tau^2(1-\tau)\int_{E^N}
\widetilde{K}^{3(N-2)}(x_{1:N},dx_{1:N}^{3(N-2)})\widehat{K}(x_{1:N}^{3(N-2)},A)\\
& \geq & \tau^4(1-\tau)^2\int_{E^N}
\widetilde{K}^{3(N-3)}(x_{1:N},dx_{1:N}^{3(N-3)})\widehat{K}^2(x_{1:N}^{3(N-3)},A)
\end{eqnarray*}
and we may thus apply the argument above, recursively, to demonstrate the result.
We now drop the
$\tau^2(1-\tau)$ and prove uniform ergodicity of $\widehat{K}$.

Let $N=2$, $A=A_1\times A_2 \in\mathcal{E}\otimes \mathcal{E}$. Using
the fact that $K_M(x_{1:2},\cdot)=K_1(x_{1},\cdot)K_2(x_{2},\cdot)$
and applying the minorization condition for $K_{2}(x_{2}^
{(2)},A_{2})$, the modified equation (\ref{peq1}) becomes
\begin{eqnarray}
\widehat{K}(x_{1:2},A) & \geq &
\int_{E^{2}}
K_{1}(x_{1},dx_{1}^{(1)})
K_{2}(x_{2},dx_{2}^{(1)}) \times{}\nonumber\\
& & {}
\int_{E^2}
\delta_{x_{1}^{(1)}}(dx_{2}^{(2)})
\delta_{x_{2}^{(1)}}(dx_{1}^{(2)})\min\bigg\{1,\frac{\pi_{1}(x_{2}^{(1)})
\pi_{2}(x_{1}^{(1)})}{\pi_{1}(x_{1}^{(1)})
\pi_{2}(x_{2}^{(1)})}\bigg\}\times{}\nonumber\\
& & {}
K_{1}(x_{1}^{(2)},A_{1})\epsilon\nu(A_{2})\label{peq2}
\end{eqnarray}
where we have ignored the rejection of an exchange move. Since
\begin{eqnarray}
\min\bigg\{1,\frac{\pi_{1}(x_{2})\pi_{2}(x_{1})}
{\pi_{1}(x_{1})\pi_{2}(x_{2})}\bigg\} & \geq &
\min\bigg\{1,\frac{\pi_{1}(x_{2})}
{\pi_{2}(x_{2})\varrho_{1}}\bigg\}~\forall \; (x_{1},x_{2})\in E^{2}
\label{peq4}
\end{eqnarray}
and using the measurability of the function, we can split
the integrals in equation (\ref{peq2}) into $I_{1}\times I_{2}$,
where:
\begin{eqnarray*}
I_{1} & = & \int_{E}K_{1}(x_{1},dx_{1}^{(1)})
\int_{E}\delta_{x_{1}^{(1)}}(dx_{2}^{(2)})\\
I_{2} & = & \int_{E^{2}}K_{2}(x_{2},dx_{2}^{(1)})
\delta_{x_{2}^{(1)}}(dx_{1}^{(2)})
\min\bigg\{1,\frac{\pi_{1}(x_{2}^{(1)})}
{\pi_{2}(x_{2}^{(1)})\varrho_{1}}\bigg\}K_{1}(x_{1}^{(2)},A_{1}).
\end{eqnarray*}
Clearly $I_{1}=1$. For $I_{2}$, integrating with respect to Dirac measure
and then applying the minorization condition we obtain:
\begin{eqnarray*}
I_{2} & \geq & \epsilon\int_{E} \nu(dx_{2}^{(1)})\min\bigg\{1,\frac{\pi_{1}(x_{2}^{(1)})}
{\pi_{2}(x_{2}^{(1)})\varrho_{1}}\bigg\}K_{1}(x_{2}^{(1)},A_{1}).
\end{eqnarray*}
Therefore equation (\ref{peq2}) becomes:
\begin{eqnarray}
\widehat{K}(x_{1:2},A) & \geq & \theta\nu^{*}(A)\label{peq3}
\end{eqnarray}
where
\begin{eqnarray*}
\theta & = & \epsilon^{2}\phi\\
\nu^{*}(A) & = & K^{*}(A_{1})\nu(A_{2})\\
K^{*}(A_{1}) & = & \frac{1}{\phi}\int_{E}
\nu(dx_{2}^{(1)})\min\bigg\{1,\frac{\pi_{1}(x_{2}^{(1)})}
{\pi_{2}(x_{2}^{(1)})\varrho_{1}}\bigg\}K_{1}(x_{2}^{(1)},A_{1})\\
\phi & = & \int_{E}
\nu(dx_{2})\min\bigg\{1,\frac{\pi_{1}(x_{2})}
{\pi_{2}(x_{2})\varrho_{1}}\bigg\}
\end{eqnarray*}
Since (\ref{peq3}) holds $\forall \; x_{1:2}\in E^{2}$, $\forall \;
A\in\mathcal{E}\otimes \mathcal{E}$ and the product
measure is non-trivial, we have that $E^{2}$ is $(1,\theta,\nu^{*})$
small, similarly $\widetilde{K}^3$ is uniformly ergodic.

Now suppose $\widehat{K}(x_{2:N},dx_{2:N})$ is uniformly
ergodic (for arbitrary mixture parameters) under the specified
conditions with $E^{N-1}$ $((N-2),\epsilon,\nu)$ small. Let
$\theta_1=\sum_{l=2}^N\varepsilon_{1l}$. For
$A=A_1\times\cdots\times A_{N}=A_{1:N}\in\mathcal{E}^N$, we have
that:
\begin{eqnarray}
\widehat{K}^{(N-1)}(x_{1:N},A) & \geq &
(1-\theta_1)^{N-2}\int_{E^N}\widehat{K}(x_{1:N},dx_{1:N}^{(1)})\widehat{K}_{\theta_1}^{(N-2)}(x_{2:N}^{(1)},A_{2:N})\times \nonumber \\
& & K_1^{2N-3}(x_1^{(1)},A_1) \label{peq3a}
\end{eqnarray}
where $\widehat{K}_{\theta_1}$ denotes that the mixture parameters
in the exchange kernel have been modified by $\theta_1$. Equation
(\ref{peq3a}) refers to the fact that the probability of moving to A
under $\widehat{K}$ for any $x\in E^N$ is greater than considering the
kernel that updates $x_1$ independently of $x_{2:N}$ (i.e.~we can
apply much the same arguments as for demonstrating
$\widetilde{K}^{3(N-1)}(x_{1:N},A) \geq
\tau^{2(N-1)}(1-\tau)^{N-1}\widehat{K}^{N-1}(x_{1:N},A)$). We then have:
\begin{eqnarray*}
\widehat{K}^{(N-1)}(x_{1:N},A) & \geq & \epsilon(1-\theta_1)^{N-2}\int_{E^N}\widehat{K}(x_{1:N},dx_{1:N}^{(1)})
\nu(A_{2:N})K_1^{2N-3}(x_1^{(1)},A_1).
\end{eqnarray*}
At this point we may apply the above arguments to yield that $E^N$ is $(N-1,\theta,\nu^*)$ small for
appropriate $\theta>0$ and probability measure $\nu^*$. Thus the result follows by induction
and the fact that $\widetilde{K}^{3(N-1)}(x_{1:N},A)
\geq \tau^{2(N-1)}(1-\tau)^{N-1}\widehat{K}^{N-1}(x_{1:N},A)$.
\end{proof}

\begin{proof}[Proof of Proposition \ref{prop1}]
Our approach is to focus upon the  Dobrushin coeficient of
$K_M\circ K_E$.
Consider $N=2$, $f:E^2\rightarrow[0,1]$:
\begin{displaymath}
|K_M\circ K_E(f)(u_{1:2}) - K_M\circ K_E(f)(v_{1:2})| \leq
\bigg|\int_{E^4}\rho_{1,2}(x_1,x_2)f(x_{1:2}')\delta_{x_1}(dx_2')\delta_{x_2}(dx_1')
K_{M}(u_{1:2},dx_{1:2})
\end{displaymath}
\begin{displaymath}
- \int_{E^4}\rho_{1,2}(x_1,x_2)f(y_{1:2}')\delta_{x_1}(dy_2')\delta_{x_2}(dy_1')K_{M}(v_{1:2},dx_{1:2})\bigg| +
\bigg|\int_{E^4}[1-\rho_{1,2}(x_1,x_2)]f(x_{1:2}')
\end{displaymath}
\begin{displaymath}
\delta_{x_1}(dx_1')\delta_{x_2}(dx_2')K_{M}(u_{1:2},dx_{1:2}) -
\int_{E^4}[1-\rho_{1,2}(x_1,x_2)]f(x_{1:2}')\delta_{x_1}(dy_1')\delta_{x_2}(dy_2')K_{M}(v_{1:2},dx_{1:2})\bigg|
\end{displaymath}
\begin{eqnarray*}
& \leq & [\textrm{osc}(f\rho_{1,2}) + \textrm{osc}(f[1-\rho_{1,2}])]\beta(K_M)\\
& \leq & 2\textrm{osc}(\rho_{1,2})\beta(K_M)
\end{eqnarray*}
where we have used the fact
that $\textrm{osc}(fg)\leq \sup(f)\textrm{osc}(g)$ (since $\inf(f)= 0$ and
$\sup(f)=1$). Note:
\begin{eqnarray*}
\textrm{osc}(\rho_{1,2}) & = & 1-\inf_{(x_1,x_2)\in E^2}\rho_{1,2}(x_1,x_2).
\end{eqnarray*}

Now, since we have the mixing condition on $K_M$ we have that:
\begin{eqnarray*}
\beta(K_M\circ K_E) & \leq & 2(1-\inf_{(x_1,x_2)\in E^2}\rho_{1,2}(x_1,x_2))(1-\epsilon).
\end{eqnarray*}
Since
\begin{eqnarray*}
\|\eta(K_M\circ K_E)^n - \pi^*\|_{TV} & \leq & \beta((K_M\circ K_E)^{n})\|\eta - \pi^*\|_{TV}
\end{eqnarray*}
(e.g.~Del Moral (2004) chapter 4)
we easily yield the desired result, for $N=2$. The result for $N\geq 3$ can
be proved using the above arguments, with only notational complications.
\end{proof}

\normalsize

\vspace{0.2 in}

{\ \nocite{*} \centerline{ REFERENCES}
\begin{list}{}{\setlength{\itemindent}{-0.3in}}
\item
{\sc Andrieu,} C. \& {\sc Moulines} E.~(2006).
On the ergodicity properties of some adaptive MCMC algorithms. \emph{Ann. Appl. Prob.}. {\bf 16}, 1462--1505.
\item
{\sc Bozdech}, Z., {\sc Llinas}, M., {\sc Pullium}, B. L.,
{\sc Wong}, E. D., {\sc Zhu}, J. \& {\sc DeRisi}, J. L. ~(2003).
The transcriptome of the intraerythrocytic developmental cycle
of \emph{Plasmodium falciparum}. {\it PLoS Biol.} {\bf 1}, 1--16.
\item
{\sc Brooks}, S. P., {\sc Giudici}, P. \& {\sc Roberts}, G. O.~(2003).
Efficient construction of reversible jump Markov chain
Monte Carlo proposal distributions (with Discussion). {\it J. R.
Statist. Soc. B} {\bf 65}, 1--37.
\item
{\sc Capp\'e}, O., {\sc Robert}, C. P. \& {\sc Ryd\'en}, T.~(2003).
Reversible jump, birth-and-death and more general continuous time
Markov chain Monte Carlo samplers. {\it J. R.
Statist. Soc. B} {\bf 65}, 679--700.
\item
{\sc Carlin}, B. \& {\sc Chib}, S.~(1995).
Bayesian model choice via Markov chain Monte Carlo methods. {\it J. R.
Statist. Soc. B} {\bf 57}, 473--84.
\item
{\sc Chauveau}, D. \& {\sc Vandekerkhove}, P.~(2002).
Improving convergence of the Hastings and Metropolis algorithm with an adaptive
proposal. {\it Scand. J. Statist.}, {\bf 29}, 13--29.
\item
{\sc Del Moral}, P.~(2004). \textit{Feynman-Kac Formulae: Genealogical and
Interacting Particle Systems with Applications}. Springer: New York.
\item
{\sc Del Moral}, P., {\sc Doucet}, A., \& {\sc Jasra}, A.~(2006).
Sequential Monte Carlo samplers.
{\it J. R. Statist. Soc. B}, \textbf{68}, 411-36.
\item
{\sc Dellaportas}, P. \& {\sc Papageorgiou}, I.~(2006).
Multivariate mixtures of normals with an unknown number of components.
{\it Statist. Comp.}, {\bf 16}, 57--68.
\item
{\sc Denison}, D. G. T., {\sc Holmes}, C. C., {\sc Mallick}, B. K.
\& {\sc Smith}, A. F. M. (2002). \textit{Bayesian Methods for
Nonlinear Classification and Regression}. Chichester: Wiley.
\item
{\sc Diaconis}, P. \& {\sc Saloff-Coste}, L.~(1993).
Comparison theorems for reversible Markov chains.
{\it Ann. Appl. Prob}, {\bf 3}, 696--750.
\item
{\sc Dobrushin}, R. L.~(1956).
Central limit theorem for non-stationary Markov chains I, II. {\it Theory Probab. Appl.}, {\bf 1}, 65--80, 329--83.
\item
{\sc Geyer}, C.~(1991),
Markov chain maximum likelihood, In
{\it Computing Science and Statistics: The 23rd symposium on the interface}, (E. Keramigas ed), 156-63
Fairfax: Interface Foundation.
\item
{\sc Geyer}, C. \& {\sc Thompson}, E. A.~(1995).
Annealing Markov chain Monte Carlo with applications
to ancestral inference. {\it J. Am.
Statist. Assoc.} {\bf 90}, 909--20.
\item
{\sc Gilks}, W. R., {\sc Roberts}, G. O.  \& {\sc George}, E. I.~(1994).
Adaptive direction sampling. {\it The Statistician} {\bf 43},
179--89.
\item
{\sc Godsill}, S.~(2001).
On the relationship between MCMC methods for uncertainty. {\it J. Comput.
Graph. Statist.} {\bf 10}, 230--48.
\item
{\sc Green}, P. J.~(1995).
Reversible jump Markov chain Monte Carlo computation
and Bayesian model determination. {\it Biometrika} {\bf 82},
711--32.
\item
{\sc Green}, P. J.~(2003a).
Discussion of Efficient construction of reversible jump Markov chain
Monte Carlo proposal distributions. {\it J. R. Statist. Soc. B}, {\bf 65}, 48-9.
\item
{\sc Green}, P. J.~(2003b).
Trans-dimensional Markov chain Monte Carlo. In
{\it Highly Structured Stochastic Systems}, (P.J. Green, N.L. Hjort \& S.
Richardson eds), 179-96
Oxford: Oxford University Press.
\item
{\sc Green}, P. J. \& {\sc Mira}, A.~(2001).
Delayed rejection in reversible jump Metropolis-Hastings. {\it Biometrika}
{\bf 88},
1035--53.
\item
{\sc Hastings}, W. K.~(1970).
Monte Carlo sampling methods using Markov chains and
their applications. {\it Biometrika} {\bf 57},
97--109.
\item
{\sc Heard}, N. A., {\sc Holmes}, C. C., \& {\sc Stephens}, D. A.~(2006).
A quantitative study of gene regulation involved
in the immune response of anopheline mosquitoes:
An application of Bayesian hierarchical clustering
of curves. {\it J. Amer. Statis. Assoc.}, {\bf 101}, 18--29.
\item
{\sc Hodgson}, M. E. A.~(1999).
A Bayesian restoration of an ion channel signal. {\it J. R.
Statist. Soc. B} {\bf 61}, 95--114.
\item
{\sc Hukushima}, K. \& {\sc Nemoto}, K.~(1996).
Exchange Monte Carlo method and application to spin
glass simulations. {\it J. Phys. Soc. Japan} {\bf 65},
1604-08.
\item
{\sc Iba}, Y.~(2001).
Extended ensemble Monte Carlo. {\it Int. J. Mod. Phys.} {\bf 12},
653-56.
\item
{\sc Jasra}, A., {\sc Holmes}, C. C., \& {\sc Stephens}, D. A.~(2005).
Markov chain Monte Carlo methods and the label switching problem in Bayesian
mixture modelling. {\it Statist.~Sci.} {\bf 20}, 50--67.
\item
{\sc Jennison}, C., {\sc Hurn}, M. A. \& {\sc Al-Awadhi}, F.~(2003).
Discussion of Efficient construction of reversible jump Markov chain
Monte Carlo proposal distributions. {\it J. R. Statist. Soc. B}, {\bf 65}, 44-5.
\item
{\sc Kou}, S. C., {\sc Zhou}, Q. \& {\sc Wong}, W.~H.~(2006).
Equi-energy sampler with applications in statistical inference and
statistical mechanics. {\it Ann. Statist.}, {\bf 34}, 1581--1619.
\item
{\sc Liang}, F. \& {\sc Wong}, W. H.~(2001).
Real parameter evolutionary Monte Carlo with applications to Bayesian mixture models. {\it J. Am.
Statist. Assoc.} {\bf 96}, 653--66.
\item
{\sc Liu}, J. S.~(2001). \textit{Monte Carlo Strategies in Scientific Computing}. New York: Springer.
\item
{\sc Madras}, N. \& {\sc Zheng}, Z.~(2003).
On the swapping algorithm. {\it Random Structures and Algorithms}
{\bf 22}, 66--97.
\item
{\sc McLachlan}, G. J. \& {\sc Peel}, D.~(2000). \textit{Finite Mixture Models}. Chichester: Wiley.
\item
{\sc Mengersen}, K. L. \& {\sc Tweedie}, R. L.~(1996).
Rates of convergence of the Hastings and Metropolis algorithms.
{\it Ann. Statist.} {\bf 24}, 101--21.
\item
{\sc Metropolis}, N., {\sc Rosenbluth}, A. W., {\sc Rosenbluth}, M. N.,
{\sc Teller}, A. H. \& {\sc Teller}, E.~(1953).
Equations of state calculations by fast computing machines. {\it J. Chem.
Phys.} {\bf 21}, 1087--92.
\item
{\sc Richardson}, S. \& {\sc Green}, P. J.~(1997).
On Bayesian analysis of mixture models with an unknown number of components
(with Discussion). {\it J. R.
Statist. Soc. B} {\bf 59}, 731--92.
\item
{\sc Robert}, C. P. \& {\sc Casella} G.~(2004). \textit{Monte Carlo Statistical
Methods}. Second edition. New York: Springer.
\item
{\sc Roberts}, G. O. \& {\sc Rosenthal}, J. S.~(2004).
General state space Markov chains and MCMC algorithms. {\it Prob. Surveys},
{\bf 1}, 20--71.
\item
{\sc Stephens}, M.~(2000).
Bayesian analysis of mixture models with an unknown number of components - an alternative to reversible jump methods. {\it Ann. Statist.} {\bf 28}, 40--74.
\item
{\sc Tierney}, L.~(1994).
Markov chains for exploring posterior distributions (with Discussion). {\it Ann. Statist.} {\bf 22}, 1701--62.
\item
{\sc Zheng}, Z.~(2003).
On swapping and simulated tempering algorithms. {\it Stoch. Proc. Appl.}
{\bf 104}, 131--53.
\end{list}
}

\end{document}